\documentclass[11pt]{article}

\usepackage{amsmath}
\usepackage{latexsym,amsmath}
\usepackage{amsthm}
\usepackage{amssymb}
\usepackage{enumerate}

\usepackage{epsfig, verbatim}
\usepackage{fullpage}
\usepackage{graphicx}
\usepackage{color}
\usepackage{tikz}
\usepackage{subfig}

\newtheorem{theorem}{Theorem}
\newtheorem{lemma}[theorem]{Lemma}

\newtheorem{corollary}[theorem]{Corollary}
\newtheorem{conjecture}[theorem]{Conjecture}

\newtheorem{claim}{Claim}

\newtheorem{remark}{Remark}

\title{Recoloring graphs via tree decompositions\thanks{The authors are supported by the ANR Grant EGOS (2012-2015) 12 JS02 002 01.}}

\author{Marthe Bonamy, Nicolas Bousquet\\ \normalsize{LIRMM, Universit\'e Montpellier 2, CNRS}\\ \small{\{marthe.bonamy, nicolas.bousquet\}@lirmm.fr}}

\begin{document}

\maketitle

\begin{abstract}
Let $k$ be an integer. Two vertex $k$-colorings of a graph are \emph{adjacent} if they differ on exactly one vertex. A graph is \emph{$k$-mixing} if any proper $k$-coloring can be transformed into any other through a sequence of adjacent proper $k$-colorings. Jerrum proved that any graph is $k$-mixing if $k$ is at least the maximum degree plus two. We first improve Jerrum's bound using the grundy number, which is the worst number of colors in a greedy coloring.

Any graph is $(tw+2)$-mixing, where $tw$ is the treewidth of the graph (Cereceda 2006). We prove that the shortest sequence between any two $(tw+2)$-colorings is at most quadratic (which is optimal up to a constant factor), a problem left open in Bonamy et al. (2012).

We also prove that given any two $(\chi(G)+1)$-colorings of a cograph (resp. distance-hereditary graph) $G$, we can find a linear (resp. quadratic) sequence between them. In both cases, the bounds cannot be improved by more than a constant factor for a fixed $\chi(G)$. The graph classes are also optimal in some sense: one of the smallest interesting superclass of distance-hereditary graphs corresponds to comparability graphs, for which no such property holds (even when relaxing the constraint on the length of the sequence). As for cographs, they are equivalently the graphs with no induced $P_4$, and there exist $P_5$-free graphs that admit no sequence between two of their $(\chi(G)+1)$-colorings.

All the proofs are constructivist and lead to polynomial-time recoloring algorithms.

\end{abstract}


\section{Introduction}

Reconfiguration problems (see~\cite{Gopalan09,ItoD11,ItoD09} for instance) consist in finding step-by-step transformations between two feasible solutions such that all intermediate results are also feasible. Such problems model dynamic situations where a given solution is in place and has to be modified, but no property disruption can be afforded. In this paper our reference problem is vertex coloring. 

In the whole paper, $G=(V,E)$ is a graph where $n$ denotes the size of $V$ and $k$ is an integer. For standard definitions and notations on graphs, we refer the reader to~\cite{Diestel}.
A \emph{(proper) $k$-coloring} of $G$ is a function $f : V(G) \rightarrow \{ 1,\ldots,k \}$ such that, for every edge $xy$, $f(x)\neq f(y)$. The \emph{chromatic number} $\chi(G)$ of a graph $G$ is the smallest $k$ such that $G$ admits a $k$-coloring.

Two $k$-colorings are \emph{adjacent} if they differ on exactly one vertex. The \emph{$k$-recoloring graph of $G$}, denoted $R_k(G)$ and defined for any $k\geq \chi(G)$, is the graph whose vertices are $k$-colorings of $G$, with the adjacency defined above. Note that two colorings equivalent up to color permutation are distinct vertices in the recoloring graph.
The graph $G$ is \emph{$k$-mixing} if $R_k(G)$ is connected. Cereceda, van den Heuvel and Johnson characterized the $3$-mixing graphs and provided an algorithm to recognize them~\cite{Cereceda09,CerecedaHJ11}. The easiest way to prove that a graph $G$ is not $k$-mixing is to exhibit a \emph{frozen} $k$-coloring of $G$, i.e. a coloring in which all vertices are adjacent to vertices of all other colors. Such a coloring is an isolated vertex in $R_k(G)$.

Deciding whether a graph is $k$-mixing is $\mathbf{PSPACE}$-complete for $k \geq 4$~\cite{BonsmaC07}. The \emph{$k$-recoloring diameter} of a $k$-mixing graph is the diameter of $R_k(G)$. In other words, it is the minimum $D$ for which any $k$-coloring can be transformed into any other through a sequence of at most $D$ adjacent $k$-colorings. The \emph{mixing number} of $G$ is the minimum integer $m(G)$ for which $G$ is $k$-mixing for every $k \geq m(G)$. It can be arbitrarily larger than the minimum $k$ for which $G$ is $k$-mixing~\cite{Cereceda} (thus arbitrarily larger than its chromatic number). Indeed, for complete bipartite graphs minus a matching, the chromatic number equals two, any 3-coloring can be transformed into any other through a linear sequence of 3-colorings if there are at least four vertices on each side of the bipartition, and the mixing number is arbitrarily large (see Fig.~\ref{fig:completmoinsmatching}). Throughout the paper, our goal is to determine bounds on the mixing number of graphs and prove that the corresponding recoloring diameter is polynomial. Note that there exists a family of graphs for which there are two colorings that can be transformed one into the other but not through a polynomial sequence~\cite{BonsmaC07}.

\begin{figure}
\centering
\includegraphics[scale=0.7,height=30mm,bb=0 0 100 100]{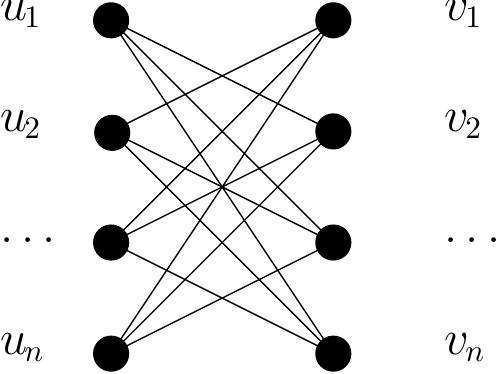}
\caption{\cite{Cereceda} A complete bipartite minus a matching. If $u_i,v_i$ are given the same color, no vertex can be recolored.}
\label{fig:completmoinsmatching}
\end{figure}

Jerrum~\cite{Jerrum95} proved that $m(G) \leq \Delta(G)+2$, where $\Delta(G)$ denotes the maximum degree of $G$. 
Let $x_1,\ldots,x_n$ be an order $\mathcal{O}$ on $V$. We denote by $N(v)$ the neighborhood of $x$. In the \emph{greedy coloring} $C(G,\mathcal{O})$ of $G$ relative to $\mathcal{O}$, every $x_i$ has the smallest color that does not appear in $N(x_i) \cap \{ x_1,\ldots,x_{i-1} \}$. Introduced in~\cite{Christen79}, the \emph{grundy number} $\chi_g(G)$ is the maximum, over all the orders $\mathcal{O}$, of the number of colors used in $C(G,\mathcal{O})$. So $\chi_g(G)$ is the worst number of colors in a greedy coloring of $G$. In Section~\ref{sect:grundy}, we prove that any graph $G$ is $(\chi_g(G)+1)$-mixing in linear time. The corresponding diameter is at most $4\cdot \chi(G)$, and we give a polynomial-time algorithm to find a sequence of length at most $2 \cdot \chi_g(G)$.

\begin{theorem}\label{thm:grundy}
For any graph $G$, if $k \geq \chi_g(G)+1$, then $G$ is $k$-mixing and the $k$-recoloring diameter is at most $4 \cdot \chi(G) \cdot n$.
\end{theorem}

Theorem~\ref{thm:grundy} improves Jerrum's bound since $\chi_g(G) \leq \Delta(G)+1$ and can be arbitrarily smaller, on stars for instance. 

Besides, the bound is tight on complete bipartite graphs minus a matching~\cite{Cereceda} (see Figure~\ref{fig:completmoinsmatching}). Indeed, $m(G)=\Delta(G)+2=\chi_g(G)+1$ since if both $u_i$ and $v_i$ are given the same color for every $i$, the resulting coloring is an $n$-coloring for which no vertex can be recolored. However, $m(G)$ is not lower-bounded by a function of $\chi_g(G)$ since for any $k$, some tree $T_k$ satisfies $\chi_g(T_k)=k$ and $m(T_k)=3$~\cite{Beyer82}.
In addition, unlike the maximum degree, the grundy number is NP-hard to compute~\cite{Zaker06}.

\medskip

Given a graph $G$ and an integer $k$, it is NP-complete to decide if $tw(G) \leq k$~\cite{ArnborgCP87}. Nevertheless, for every fixed $k$, there is a linear-time algorithm to decide if the treewidth is at most $k$ (and find a tree decomposition)~\cite{Bodlaender93}. Graphs of treewidth $k$, being $k$-degenerate, are $(k+2)$-mixing~\cite{Cereceda}. However, the best upper-bound known on the recoloring diameter is exponential. In Section~\ref{sect:prooftw}, we prove that the recoloring diameter is polynomial for bounded treewidth graphs, a problem left open in~\cite{BonamyJ12}.

\begin{theorem}\label{thm:tw}
For every graph $G$, if $k \geq tw(G)+2$, then $G$ is $k$-mixing and its $k$-recoloring diameter is at most $2 \cdot (n^2+n)$.
\end{theorem}

Section~\ref{sect:prooftw} is devoted to a proof of Theorem~\ref{thm:tw}, which is independent from Theorem~\ref{thm:grundy}. The quadratic upper bound on the recoloring diameter was known for chordal graphs~\cite{BonamyJ12}, but its generalization to bounded treewidth graphs was left open. As shown in the case of chordal graphs~\cite{BonamyJ12} (which is a subclass of graphs of treewidth $\chi(G)-1$), the mixing number is tight, and the recoloring diameter is tight up to a constant factor. The existence of a polynomial recoloring diameter in the case of $k$-degenerate graph is still open and will be discussed in Section~\ref{sec:ccl}. A sketch of the proofs of Theorem~\ref{thm:grundy} and~\ref{thm:tw} were published in the extended abstract of LAGOS'13~\cite{BonamyB13}.

\medskip

The last main results of this paper deal with cographs and distance-hereditary graphs. The class of cographs is the class of $P_4$-induced free graphs, i.e. the class of graphs that do not contain a path on four vertices as an induced subgraph. By Theorem~\ref{thm:grundy}, we obtain the following:

\begin{corollary}\label{thm:cographs}
For every cograph $G$, if $k \geq \chi(G)+1$, then $G$ is $k$-mixing and its $k$-recoloring diameter is at most $\mathcal{O}(\chi(G) \cdot n)$.
\end{corollary}

Indeed, the grundy number of cographs is equal to their chromatic number~\cite{Gyarfas88}. Recall that cographs are $P_4$-free graphs. Can we generalize this result to the class of $P_k$-free graphs for any fixed $k$? It is easy to answer negatively to this question by observing that complete bipartite graphs minus a matching (see Fig.~\ref{fig:completmoinsmatching}) are $P_6$-free graphs and that the mixing number of such graphs is arbitrarily large. We can build an ad-hoc $P_5$-free graph with chromatic number $4$ which is not $5$-mixing (see Fig.~\ref{fig:P5free}). In fact, we can even build a family $(G_k)_{k \geq 3}$ of $P_5$-free graphs that admit both a $(k+1)$-coloring and a frozen $2k$-coloring, as follows. Take a clique $\{u_1,u_2,\ldots,u_{2k}\}$, remove all edges $(u_i,u_{k+i})$, add for all $i$ a vertex $v_i$ adjacent to all the $u_j$'s except $u_i$. We obtain a $(k+1)$-coloring $\alpha_k$ by setting $\alpha_k(u_i)=\alpha_k(u_{k+i})=i$ for all $1 \leq i \leq k$ and $\alpha_k(v_i)=k+1$ for all $i$ such that $v_i$ exists, and a frozen $2k$-coloring $\beta_k$ by setting $\beta_k(u_i)=i$ and $\beta_k(v_i)=\beta_k(u_i)$ for all $1 \leq i \leq 2k$. Not all edges are necessary, but it is easier that way to verify that there is indeed no induced $P_5$.

Note that, since the $3$-recoloring diameter of the class of induced paths is quadratic~\cite{BonamyJ12}, only graph classes excluding long paths can have a linear recoloring diameter. Corollary~\ref{thm:cographs} ensures that $P_4$-free graphs admit a linear recoloring diameter. So we raise the following question: does the class of $P_k$-free graphs with mixing number $\chi(G)+1$ have a linear recoloring diameter?

\begin{figure}
\centering
\subfloat[][A $4$-coloring.]{
\centering
\begin{tikzpicture}[scale=1.5]
\tikzstyle{blacknode}=[draw,circle,fill=black,minimum size=6pt,inner sep=0pt]

\draw (0,-1) node[blacknode] (u1) [label=90:$2$] {};
\foreach \i in {2,3,4,5,6}
{\pgfmathtruncatemacro{\j}{\i - 1}
\pgfmathtruncatemacro{\color}{1+mod(\i,3)}
\draw (u\j)
-- ++(-30+\j*60:1cm) node[blacknode] (u\i) [label=180+-90+\j*60:$\color$] {};}
\draw (u1) edge node  {} (u6);

\foreach \i/\angle in {1,2,3,4,5,6}
{\pgfmathtruncatemacro{\j}{1+mod(\i,6)}
\pgfmathtruncatemacro{\k}{1+mod(\i+1,6)}
\draw (u\i) edge node  {} (u\j);
\draw (u\i) edge node  {} (u\k);
}

\foreach \i/\angle in {1/-60,3/60,4/120,6/-120}
{
\pgfmathtruncatemacro{\j}{1+mod(\i+4,6)}
\pgfmathtruncatemacro{\k}{1+mod(\i,6)}
\pgfmathtruncatemacro{\l}{1+mod(\i+1,6)}
\draw(\angle:1.7) node[blacknode] (v\i) [label=\angle: $4$] {};
\draw (v\i) edge [bend left] node  {} (u\i);
\draw (v\i) edge [bend left] node  {} (u\j);
\draw (v\i) edge [bend right] node  {} (u\k);
\draw (v\i) edge [bend right] node  {} (u\l);
}
\end{tikzpicture}}
\qquad
\subfloat[][A frozen $5$-coloring.]{
\centering
\begin{tikzpicture}[scale=1.5]
\tikzstyle{blacknode}=[draw,circle,fill=black,minimum size=6pt,inner sep=0pt]

\draw (0,-1) node[blacknode] (u1) [label=90:$2$] {};
\foreach \i/\color in {2/3,3/4,4/2,5/5,6/1}
{\pgfmathtruncatemacro{\j}{\i - 1}
\draw (u\j)
-- ++(-30+\j*60:1cm) node[blacknode] (u\i) [label=180+-90+\j*60:$\color$] {};}
\draw (u1) edge node  {} (u6);

\foreach \i/\angle in {1,2,3,4,5,6}
{\pgfmathtruncatemacro{\j}{1+mod(\i,6)}
\pgfmathtruncatemacro{\k}{1+mod(\i+1,6)}
\draw (u\i) edge node  {} (u\j);
\draw (u\i) edge node  {} (u\k);
}

\foreach \i/\angle/\color in {1/-60/5,3/60/1,4/120/3,6/-120/4}
{
\pgfmathtruncatemacro{\j}{1+mod(\i+4,6)}
\pgfmathtruncatemacro{\k}{1+mod(\i,6)}
\pgfmathtruncatemacro{\l}{1+mod(\i+1,6)}
\draw(\angle:1.7) node[blacknode] (v\i) [label=\angle: $\color$] {};
\draw (v\i) edge [bend left] node  {} (u\i);
\draw (v\i) edge [bend left] node  {} (u\j);
\draw (v\i) edge [bend right] node  {} (u\k);
\draw (v\i) edge [bend right] node  {} (u\l);
}
\end{tikzpicture}}
\caption{A $P_5$-free graph that is $4$-colorable and not $5$-mixing.}
\label{fig:P5free}
\end{figure}
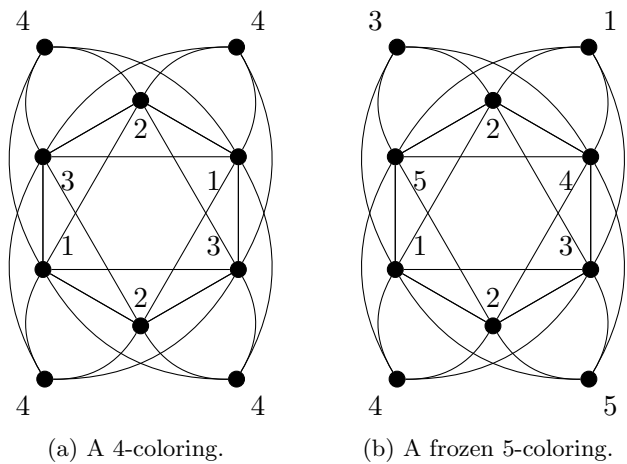

Using a slight generalization of cographs, we finally prove in Section~\ref{sect:CDHproof} that distance-hereditary graphs admit a quadratic recoloring diameter.
\begin{theorem}\label{thm:distancehereditary}
For every distance-hereditary graph $G$ and every $k \geq \chi(G)+1$, the graph $G$ is $k$-mixing and it recoloring diameter is at most $\mathcal{O}(k \cdot \chi(G) \cdot n^2)$.
\end{theorem}
Once again, since distance-hereditary graphs contain long paths, the upper bound on the recoloring diameter is optimal up to a constant factor.
Since the mixing number of complete bipartite graph minus a matching is not bounded by a function of the chromatic number, neither is it the case for any graph class containing them, thus for comparability graphs, perfectly orderable graphs, and then perfect graphs, answering a question of \cite{BonamyJ12}. Considering that comparability graphs form the smallest interesting superclass of distance-hereditary graphs, Theorem~\ref{thm:distancehereditary} is in some sense optimal with regards to the graph class, as it is optimal as to the bound on the recoloring diameter.

Finally note that all our results are also algorithmic. Indeed each proof can be adapted into a polynomial-time algorithm that will find a short recoloring sequence. The one exception is for bounded treewidth graphs. Indeed determining the treewidth of a graph is an $NP$-hard problem. Nevertheless, for every fixed $k$, deciding whether the treewidth of the graph is at most $k$ (and if so, producing a corresponding tree decomposition) can be done in linear time~\cite{Bodlaender93}.

\section{Preliminaries}
Let us first recall some classical definitions on sets. Let $X$ and $Y$ be two subsets of $V$. The set $X \setminus Y$ is the subset of elements $x \in X$ such that $x \notin Y$. By abuse of notation, given a set $X$ and an element $x$, $X\setminus x$ denotes $X \setminus \{ x\}$ and $\{x\}$ will sometimes be denoted by $x$. The \emph{size} $|X|$ of $X$ is its number of elements. 

Let $G=(V,E)$ be a graph. For any coloring $\alpha$ of $G$, we denote by $\alpha(H)$ the set of colors used by $\alpha$ on the subgraph $H$ of $G$. The \emph{neighborhood} of a vertex $x$, denoted by $N(x)$ is the subset of vertices $y$ such that $xy \in E$. The \emph{length} of a path is its number of edges. The \emph{distance} between two vertices $x$ and $y$, denoted $d(x,y)$, is the minimum length of a path between these two vertices. When there is no path, the distance is infinite. The \emph{distance} between two $k$-colorings of $G$ is implicitely the distance between them in the recoloring graph $R_k(G)$. Let us first recall a classical result on recoloring.

\begin{lemma}\label{lemma:clique}
If $k \geq n+1$, any $k$-coloring of $K_n$ can be transformed into any other by recoloring every vertex at most twice.
\end{lemma}
\begin{proof}
Let $\alpha, \beta$ be two colorings of $K_n$. Assume $K_n$ is initially colored in $\alpha$. We plan to recolor $K_n$ into $\beta$ while recoloring each vertex at most twice. Let $D$ be the digraph on $n$ vertices with an arc $xy$ if $y$ is currently colored in $\beta(x)$. Informally $xy$ is an arc if the current color of $y$ prevents the recoloring of $x$ into its final color (in $\beta$). The coloring of $K_n$ is proper at any time: no two vertices are colored identically, so $d^+(x) \leq 1$. The same argument holds for $\beta$, so $d^-(x) \leq 1$. Hence $D$ is a disjoint union of directed paths and of circuits. 

We recolor every directed path as follows. Let $x_0,x_1,\ldots,x_k$ be a directed path. Then $d^+(x_k)=0$, and we can recolor $x_k$ into its final color in $\beta$. Then $d^+(x_{k-1})=0$, and we recolor $x_{k-1}$ into its final in $\beta$. We repeat that operation until $x_0$ is recolored. Now every vertex of the directed path is an isolated vertex in $D$, since no two vertices are colored identically in $\beta$.

Let $x_0,x_1,\ldots,x_k,x_0$ be a circuit. Since $k \geq n+1$, $x_0$ can be recolored with a free color. After this recoloring, we have $d^+(x_k)=0$, so $x_0,\ldots,x_k$ becomes a directed path, which we recolor as such. We then recolor $x_0$ into $\beta(x_0)$. Now every vertex of circuit is an isolated vertex in $D$. We repeat that operation until no circuit remains in $D$. Since we already recolored every directed path and $D$ was initially a disjoint union of directed paths and circuits, the clique $K_n$ is now colored in $\beta$.
\end{proof}

\section{Mixing number and grundy number}\label{sect:grundy}

This section is devoted to a proof of Theorem~\ref{thm:grundy}, which is derived from the following lemma.
\begin{lemma}\label{lem:grundystep}
Let $G$ be a graph on $n$ vertices, and $k \geq \chi_g(G)+1 > \ell$. For any $k$-coloring $\alpha$ of $G$ and any grundy $\ell$-coloring $\beta$ of $G$, we have $d(\alpha,\beta)\leq (2 \cdot \ell-1) \cdot n$.
\end{lemma}
\begin{proof}
Let us prove it by induction on $\ell$.
If $\ell=1$, $G$ has no edge and in $n$ steps, we can transform $\alpha$ into $\beta$.
Assume now that $\ell\geq 2$. For any integer $i$ and any coloring $c$, $V_i^c$ is the set of vertices of color $i$ in $c$. Iteratively on $i$ from $1$ to $\ell$, we recolor the vertices of $V^{\alpha}_i$ with the smallest color for which the coloring is still proper. The resulting coloring $\gamma$ of $G$ is the greedy coloring relative to the order $V^{\alpha}_1, V^{\alpha}_2, \ldots, V^{\alpha}_{\ell}$. Hence $\gamma$ is an (at most) $\chi_g(G)$-coloring. In addition, $d(\alpha, \gamma) \leq n$, since no vertex is recolored twice. 
Since no vertex is colored with $\chi_g(G)+1$ in $\gamma$ and $k \geq \chi_g(G)+1$, we recolor vertices of $V_1^{\gamma}\setminus V_1^{\beta}$ with color $\chi_g(G)+1$. We then recolor vertices of $V_1^{\beta}$ with $1$ if needed. The resulting coloring $\delta$ satisfies $V_1^{\delta}=V_1^{\beta}$. In addition, $d(\gamma,\delta) \leq n$, since no vertex is recolored twice. 
Let us now prove that the induction hypothesis holds on $G'=G(V \setminus V_1^{\beta})$ with $k-1$, $\ell-1$, $\delta_{G'}$ and $\beta_{G'}$. We have $\chi_g(G') < \chi_g(G)$. Indeed, assume that there is an order $\mathcal{O}$ on $V \setminus V_1^{\beta}$ such that $\chi_g(G')= \chi_g(G)$. Consider the order $\mathcal{O}'=(V_1^{\beta},\mathcal{O})$ on $V$. Every vertex of $\mathcal{O}$ has a neighbor on $V_1^{\beta}$ (since $\beta$ is grundy), so the greedy coloring relative to $\mathcal{O}'$ needs $\chi_g(G)+1$ colors for $G$, a contradiction. We also have $\ell-1 \leq \chi_g(G')$. Indeed, for $\mathcal{O}$ the order on $V$ corresponding to $\beta$, the greedy coloring of $G'$ relative to $\mathcal{O}\setminus V_1^{\beta}$ requires at least $\ell-1$ colors. So we can apply the induction hypothesis on $G'$ with $k-1$, $\ell-1$, $\delta_{G'}$ and $\beta_{G'}$ (the color $1$ is forgotten, thus a greedy coloring will start with color $2$, and $\beta_{G'}$ is an $\ell-1$ grundy coloring of $G'$). This ensures that $G'$ can be recolored in $(2 \cdot (\ell-1)-1) \cdot |V(G')| \leq (2 \cdot (\ell-1) -1) \cdot n$ steps.
Consequently, $d(\alpha,\beta)\leq d(\alpha,\gamma)+d(\gamma,\delta)+d(\delta,\beta) \leq (2 \cdot \ell -1) \cdot n$.
\end{proof}
Finally, we consider a graph $G$ on $n$ vertices, an integer $k \geq \chi_g(G)+1$, and two $k$-colorings $\alpha_1$ and $\alpha_2$ of $G$. By definition of the chromatic number, we know there exists a grundy $\chi(G)$-coloring $\beta$ of $G$. Since, $d(\alpha_1,\alpha_2) \leq d(\alpha_1,\beta)+d(\beta,\alpha_2)$, Lemma~\ref{lem:grundystep} ensures that $d(\alpha_1,\alpha_2) \leq 4 \cdot \chi(G) \cdot n$. However, finding a grundy $\chi(G)$-coloring of $G$ is NP-complete. For a polynomial-time algorithm, we can derive from $\alpha_2$ a grundy coloring $\delta$ in at most $n$ steps, and apply Lemma~\ref{lem:grundystep} on $\alpha_1$ and $\delta$, to obtain that $d(\alpha_1,\alpha_2) \leq 2 \cdot \chi_g(G) \cdot n$.

\section{Bounded treewidth graphs}\label{sect:prooftw}

The aim of this section is to prove Theorem~\ref{thm:tw}. A \emph{tree} is a connected graph without cycles. To avoid confusion, its vertices are called \emph{nodes}. A \emph{tree decomposition} of $G$ is a tree $T$ such that:
\begin{itemize}
 \item To every node $u$ of $T$, we associate a \emph{bag} $B_u \subseteq V$.
 \item For every edge $xy$ of $G$, there is a node $u$ of $T$ such that both $x$ and $y$ are in $B_u$.
 \item For every vertex $x \in V$, the set of nodes of $T$ whose bags contain $x$ forms a non-empty subtree in $T$.
\end{itemize}
The \emph{size} of a tree decomposition $T$ is the largest number of vertices in a bag of $T$, minus one. The \emph{treewidth} $tw(G)$ of $G$ is the minimum size of a tree decomposition of $G$.

A \emph{chordal graph} is a graph that admits a perfect elimination ordering: that is, the vertices of the graphs can be ordered $v_1,v_2,\cdots,v_p$ in such a way that the neighborhood of any vertex $v_i$ in $\{v_1,v_2,\cdots,v_{i-1}\}$ forms a clique. Any chordal graph $G$ admits a tree decomposition whose bags are the maximal cliques of $G$.

Actually, any tree decomposition of a graph $G$ can be viewed as a chordal graph $H$ with vertex set $V(G)$ that admits $G$ as a subgraph ($H$ is a \emph{supergraph} of $G$). 
Informally, we transform step-by-step any $(tw+2)$-coloring of a graph into a $(tw+2)$-coloring of a ''good'' chordal supergraph with the same treewidth.

We first introduce particular tree decompositions, called complete tree decompositions. In such decompositions, all the bags have exactly the same size and any two adjacent bags differ on exactly one vertex. Two vertices are parents if their subtrees are, in some sense, adjacent. A $V$-coherent coloring is a coloring where parents are colored identically, i.e. a coloring of the vertices that is compatible with the chordal supergraph corresponding to the complete tree decomposition.

The proof is divided into two parts. First we prove that the distance between $V$-coherent colorings is linear. We then prove that any coloring can be transformed into a $V$-coherent coloring with a quadratic number of recoloring steps as long as the number of colors is at least $tw(G)+2$.

\subsection{Families}

\begin{figure}
 \centering
  \includegraphics[scale=1,height=40mm]{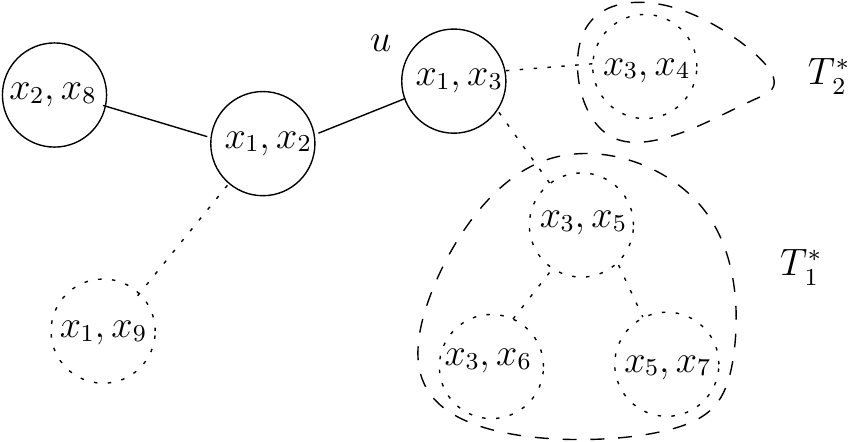}
  \caption{A $1$-complete tree decomposition $T$.}
  \label{fig:completedecomp}
\end{figure}

A tree decomposition $T$ of a graph $G$ is \textit{$\ell$-complete} when every bag has size $\ell+1$ and any two adjacent nodes $u,v$ satisfy $|B_u \cap B_v| = \ell$. In other words, for every edge $uv$ of $T$, there exists a vertex $x$ such that $x = B_u \setminus B_v$. For any subtree $T'$ of $T$, $B_{T'}$ denotes $\cup_{v \in T'} B_v$. Let $X \subseteq V$. The tree decomposition $T[V \setminus X]$ is the same tree as $T$ except that the bag of every node $u$ is $B_u \setminus X$, and that every edge $uv$ of $T$ is contracted if $B_u\setminus X \subseteq B_v \setminus X$. In Fig.~\ref{fig:completedecomp}, the full-line edges subtree is $T[V \setminus \{ x_4,x_5,x_6,x_7,x_9 \}]$. The following remark is an immediate consequence of the definition of complete tree decomposition.

\begin{remark}\label{rem:subtreelcomplete}
Any connected subtree of an $\ell$-complete tree decomposition is still $\ell$-complete.
\end{remark}

A \emph{baby} is a vertex of $V$ that appears in exactly one bag $B_u$, where $u$ is a leaf of $T$. Note that all the neighbors of a baby $x$ are in $B_u$. In Fig.~\ref{fig:completedecomp}, vertex $x_8$ is a baby.

\begin{remark}\label{rem:lcomplete}
Let $T$ be an $\ell$-complete tree decomposition. If $x$ is a baby then $T[V \setminus x]$ is $\ell$-complete.
\end{remark}
\begin{proof}
Let $u$ be the unique node whose bag contains $x$. Then the only modified bag in $T[V \setminus x]$ is $B_u$. Let $v$ be the father of $u$ in $T$. Since $T$ is complete, $B_u \setminus B_v = x$ in $T$, so the edge $uv$ is contracted in $T[V \setminus x]$. Therefore $T[V \setminus x]$ is exactly $T \setminus u$ which is $\ell$-complete by Remark~\ref{rem:subtreelcomplete}.
\end{proof}

We first prove that every graph admits complete tree decompositions. Then we derive from it the notion of parents and family between vertices of $G$.

\begin{lemma}\label{completetree}
For every graph $G$, if $n-1 \geq \ell \geq tw(G)$ then $G$ admits an $\ell$-complete tree decomposition.
\end{lemma}

\begin{proof}
By definition of $tw(G)$, the graph $G$ admits a tree decomposition $T$ whose every bag has size at most $\ell+1$. We can assume w.l.o.g. that no bag is contained in another in $T$: if two adjacent nodes $u,v$ in $T$ verify $B_u \subseteq B_v$, then the edge $uv$ can be contracted.

We build inductively an $\ell$-complete tree decomposition $T_c$ of $G$ such that every bag of $T$ is contained in a bag of $T_c$.

If $n=\ell+1$, then the tree decomposition consisting of a single node with bag $V(G)$ is $\ell$-complete.

If $n\geq \ell+2$, then $T$ has at least two nodes since every vertex is contained in at least one bag. Let $u$ be a leaf of $T$ and $v$ be the neighbor of $u$. Since no bag is contained in another in $T$, there is a vertex $x$ in $B_u \setminus B_v$. Note that $x$ is a baby. Otherwise the subset of nodes whose bags contain $x$ would not be a subtree of $T$ since $x\notin B_v$ and $v$ is the unique neighbor of $u$ in $T$. Let $T'=T[V \setminus x]$.

By induction hypothesis, $G \setminus x$ admits an $\ell$-complete tree decomposition $T'_c$ where every bag of $T'$ is contained in a bag of $T'_c$. So some node $w$ of $T'_c$ satisfies $(B_u \setminus x) \subseteq B_w'$. Since $|B'_w|= \ell+1 \geq tw(G)+1$, some vertex $y$ of $B_w'$ is not in $B_u$. We consider $T_c$ built from $T'_c$ by adding a leaf $u'$ attached on $w$ whose bag is $(B_w' \cup  x ) \setminus y$. Then $T_c$ is an $\ell$-complete tree decomposition of $G$ with the required property with regards to $T$.
\end{proof}

Let $T$ be a complete tree decomposition. Note that $|B_u \setminus B_v|=|B_v \setminus B_u|=1$ for every edge $uv$. Two vertices $x,y \in V$ are \textit{$T$-parents} if there are two adjacent nodes $u,v$ of $T$ such that $x=B_u \setminus B_v$, and $y=B_v \setminus B_u$. In other words, vertices $x$ and $y$ are $T$-parents if the subtree of the nodes containing $x$ in their bags and the subtree of the nodes containing $y$ in their bags do not intersect, but are connected by an edge ($uv$ in this case). When no confusion is possible we will use the term parents instead of $T$-parents. Also note that the notion of parents is symmetric: if $x$ is a parent of $y$ then $y$ is a parent of $x$.
The \emph{family relation} is the transitive closure of the parent relation. A \emph{family} is a class of the family relation. In Fig.~\ref{fig:completedecomp}, the families are $\{x_1,x_4,x_5,x_6,x_8\}$ and $\{x_2,x_3,x_7,x_9\}$. The partition of $V$ induced by the families is called the \emph{family partition}. In Fig.~\ref{fig:completedecomp}, vertices $x_2$ and $x_3$ are parents.

\begin{remark}\label{rem:family}
The family partition of any $\ell$-complete tree decomposition exists and is unique. Each family contains exactly one vertex in every bag. So there are $\ell+1$ families, which are stable sets.
\end{remark}

\begin{proof}
By induction on $T$. If $T$ has a single node $u$, then no vertex has a parent. So each family is a single vertex. 
Assume $T$ has at least two nodes. Let $u$ be a leaf of $T$ and $v$ be its adjacent node. Note that the family partitions of $T$ are the extensions of those of $T \setminus u$. The vertices $x=B_u \setminus B_v$ and $y=B_v \setminus B_u$ are parents and $y$ is the unique parent of $x$.
Since $u$ is a leaf of $T$, $T\setminus u$ is still $\ell$-complete by Remark~\ref{rem:subtreelcomplete}. 

By induction, $B_v$ contains exactly one vertex of every family of the unique family partition of $T[V \setminus u]$. Since $B_u =(B_v \cup x) \setminus y$, and since $y$ is the unique parent of $x$, we can uniquely extend the partition by adding $x$ in the family of $y$. Besides, in $B_u$ there is exactly one vertex of each family.
\end{proof}

\subsection{Coherent colorings}

Let $T$ be an $\ell$-complete tree decomposition of a graph $G$. A coloring $\alpha$ is \emph{$X$-coherent} (relatively to $T$) if for every $x,y \in X$ that are parents, $\alpha(x)=\alpha(y)$ and for every bag $B$ and every $x \in X\cap B$, only $x$ is colored with $\alpha(x)$ in $B$.
Note that since parents are non-adjacent in the graph by Remark~\ref{rem:family}, coherent colorings can be proper. Note that a $V(G)$-coherent coloring is a proper $\ell$-coloring which is in some sense canonical: given two $V$-coherent colorings, they differ only up to a color permutation. Our recoloring algorithm consists in transforming any coloring into such a canonical coloring. Then we can transform any $V$-coherent coloring into any other using the recoloring algorithm of the clique.

The subsection is organized as follows. First we define the notion of merging. Then we prove that the distance between $V(G)$-coherent colorings is linear. And we finally provide some recoloring lemmas regarding $(V \setminus B_u)$-coherent colorings. All these tools will be used in Section~\ref{sec:obtainvcoherent}.

Let $G$ be a graph and $\mathcal{C}$ be a stable set. The \emph{merged graph} on $\mathcal{C}$ is the graph $G$ where vertices of $\mathcal{C}$ are identified into a vertex $z$ and $xz$ is an edge if there exists a vertex $y \in \mathcal{C}$ such that $xy$ is an edge. A coloring $\gamma$ of the merged graph can be \emph{extended} on the whole graph by coloring every vertex of $\mathcal{C}$ with $\gamma(z)$. For any stable sets $\mathcal{C}_1,\mathcal{C}_2,\cdots,\mathcal{C}_p$ with $\mathcal{C}_i \cap \mathcal{C}_j = \emptyset$ for any $i \neq j$, the \emph{merged graph} on $\mathcal{C}_1,\mathcal{C}_2,\cdots,\mathcal{C}_p$ is the graph obtained from $G$ by merging successively $\mathcal{C}_1, \mathcal{C}_2 \cdots \mathcal{C}_p$.

\begin{remark}\label{rem:merged}
Let $\mathcal{C}$ be a stable set. Let $\alpha', \beta'$ be two colorings of the merged graph on $\mathcal{C}$ and $\alpha,\beta$ be their extended colorings. If $\alpha'$ can be transformed into $\beta'$ by recoloring each vertex at most $t$ times, then $\alpha$ can be transformed into $\beta$ by recoloring every vertex at most $t$ times.
\end{remark}
\begin{proof}
We just have to follow the recoloring process of $\alpha'$ into $\beta'$. If the recolored vertex is not $\mathcal{C}$, then perform the same recoloring in the extended graph. Otherwise, recolor the vertices of $\mathcal{C}$ (one after the other) into the new color of $\mathcal{C}$ in the extended graph. All the intermediate colorings are proper since $\mathcal{C}$ is a stable set.
\end{proof}
Remark~\ref{rem:merged} formalizes an easy fact: when there is a set of a same color, then we can consider it as a single vertex.

\begin{remark}\label{rem:mergeddec}
Let $T$ be an $\ell$-complete tree decomposition, let $C$ be a stable set of $G$ that belongs to the same family and let $G'$ be the merged graph on $C$. If $T[C]$ is connected, then for $T'$ the tree obtained from $T$ by contracting any edge $uv$ such that $B_u$ and $B_v$ differ only on vertices of $C$, $T'$ is an $\ell$-complete tree decomposition of $G'$.
\end{remark}
\begin{proof}
By induction on $|C|$.

If $|C|=1$, $G'$ and $T'$ are actually $G$ and $T$, so the result holds.

If $|C|=2$, since $T[C]$ is connected, for $C=\{a,b\}$, $a$ and $b$ are parents. Let $G'$ be the merged graph on $\{a,b\}$ and $T'$ be the tree obtained from $T$ by contracting any edge $uv$ such that $B_u$ and $B_v$ differ only on vertices of $\{a,b\}$. The set of nodes of $T$ whose bags contain $a$ or $b$ forms a non-empty subtree of $T$, so the same holds of $T'$. Then the tree $T'$ is a tree decomposition of $G'$. The tree decomposition is still $\ell$-complete since every bag of $T'$ corresponds to (at least) a bag of $T$, and any two adjacent nodes in $T'$ correspond to two adjacent nodes in $T$.

If $|C|\geq 3$, we consider a leaf $a$ of $T[C]$. We apply the induction hypothesis with $C \setminus \{a\}$: we obtain the graph $G'$ and the $\ell$-complete tree decomposition $T'$. Let $b$ be the vertex of $G'$ corresponding to $C \setminus \{a\}$. Note that $b$ belongs to the same family as $a$. Since $T[C]$ is connected, so is $T'[\{a,b\}]$. We apply the induction hypothesis on $G'$ and $T'$ with $\{a,b\}$. The resulting graph $G''$ and $\ell$-complete tree decomposition $T''$ are also the graph and tree decomposition that would have been obtained by merging directly on $C$.
\end{proof}

\begin{lemma}\label{compatibleson}
Let $k \geq tw(G)+2$. If every $k$-coloring of $G$ can be transformed into a $V$-coherent coloring with at most $f(n)$ recolorings, then the $k$-recoloring diameter of $G$ is at most $2\cdot(f(n)+n)$.
\end{lemma}

\begin{proof}
Let $\alpha, \beta$ be two $k$-colorings of $G$. By assumption, there are two $V$-coherent colorings $\gamma_\alpha$ and $\gamma_\beta$ such that $d(\alpha,\gamma_\alpha) \leq f(n)$ and $d(\beta,\gamma_\beta) \leq f(n)$. 

Let us prove that $d(\gamma_\alpha,\gamma_\beta)\leq 2n$. By definition, all the vertices of a same family are colored identically in $\gamma_\alpha$. The same holds for $\gamma_\beta$. Let $G'$ be the merged graph where every family is identified into a same vertex. By Remark~\ref{rem:family}, the family partition is unique, so both $\gamma_\alpha$ and $\gamma_\beta$ are extensions of $\gamma_\alpha'$ and $\gamma_\beta'$ colorings of $G'$. Every pair of vertices of $G'$ have distinct colors in $\gamma_\alpha$ (and in $\gamma_\beta$). So $G'$ can be considered as a clique on $tw(G)+1$ vertices (since there are $tw(G)+1$ families). 
Lemma~\ref{lemma:clique} and Remark~\ref{rem:merged} ensure that $d(\gamma_\alpha,\gamma_\beta)\leq 2n$.
Since $d(\alpha,\beta) \leq d(\alpha,\gamma_\alpha)+d(\gamma_\alpha,\gamma_\beta)+d(\gamma_\beta,\beta)$, Lemma~\ref{compatibleson} holds.
\end{proof}

Let us first make some observations for the two forthcoming lemmas. Let $T$ be a tree and $u$ be a node of $T$. We can consider that $T$ is rooted on $u$. Then $w$ is the \emph{father of $v$} if $vw$ is an edge and $v$ is not in the connected component of $u$ in $T \setminus w$. The \emph{tree rooted on $v$}, denoted by $T_v$, is the connected component of $v$ in $T \setminus w$. Note that if $T$ is an $\ell$-complete tree decomposition, then so is $T_v$ for any $v$. Let us first prove some stability on $(V \setminus B_u)$-coherent colorings. This slightly technical lemma is at the core of the recoloring algorithm.

\begin{lemma}\label{claim:everywhere}
Let $T$ be an $\ell$-complete tree decomposition rooted at $u$ and let $v$ be a node of $T$ distinct from $u$. Let $\alpha$ be a $(V \setminus B_u)$-coherent coloring where color $a$ does not appear in $B_u$.\\
If a vertex of $B_v$ is colored with $a$, every bag of $T_v$ contains a vertex colored with $a$.
\end{lemma}
\begin{proof}
Assume by contradiction that a node $w$ of $T_v$ does not contain a vertex colored with $a$ in its bag. Choose $w$ in such a way $w$ is as near as possible from $v$ in $T$. Then the father $w'$ of $w$ contains a vertex $y$ of color $a$. 
The vertex $y$ is not in $B_u$ since $\alpha(y)=a$. Let $z= B_w \setminus B_{w'}$. We have $z \notin B_u$ since $z \notin B_{w'}$ and $w'$ is the father of $w$. So $y$ and $z$ are parents. Since $\alpha$ is $(V\setminus B_u)$-coherent, we have $\alpha(y)=\alpha(z)$. But $\alpha(y)=a$, a contradiction.
\end{proof}

\begin{lemma}\label{pfclaim1}
Let $k,\ell$ be two integers with $k \geq \ell+2$.
Let $T$ be an $\ell$-complete tree decomposition rooted on $u$. Let $\alpha$ be a $(V \setminus B_u)$-coherent $k$-coloring where some color $a$ does not appear in $B_u$. 

Then by recoloring every vertex of $V \setminus B_u$ at most once, we can obtain a $(V \setminus B_u)$-coherent $k$-coloring where no vertex is colored with $a$.
\end{lemma}
\begin{proof}
Let us prove it by induction on $|T|$. Given a subset of vertices $X$ of $V$, two vertices are in the same \emph{$X$-family} if they are in the closure relation of the parents relation restricted to the vertices of $X$ (note that if $X=V$ the two notions coincide). For every ($V \setminus B_u$)-family $C$, we merge all the vertices of $C$. Since the coloring is $(V \setminus B_u)$-coherent, it makes sense: all the vertices of $C$ were colored identically. By Remark~\ref{rem:mergeddec}, the resulting tree decomposition is still $\ell$-complete. By Remark~\ref{rem:merged}, recoloring the resulting graph suffices: the resulting coloring will be $(V\setminus B_u)$-coherent in any case.

If $\ell=0$, then the graph has no edge. Color $a$ can be eliminated by recoloring every vertex at most once since the graph is a stable set.

If $\ell >0$, for every son $v$ of $u$, we proceed as follows. If a vertex of $B_v$ is colored with $a$ then it is necessarily the vertex $B_v \setminus B_u = y$. We consider two cases depending on whether $\alpha(y)=a$.
\begin{itemize}
\item Assume $\alpha(y) \neq a$. Since $T_v$ is smaller than $T$, by the induction hypothesis we can recolor $B_{T_v}$ without recoloring any vertex of $B_v$ (and then of $B_u$), so as to obtain a $(B_{T_v}\setminus B_v)$-coherent coloring with no vertex of color $a$.
\item Assume $\alpha(y)=a$. By Remark~\ref{rem:family}, every node of $T_v$ contains exactly one vertex of the same family as $y$. Since we merged on the families of $V\setminus B_u$, vertex $y$ belongs to every node of $T_v$. Thus $T_v[B_{T_v}\setminus\{y\}]$ is an $(\ell-1)$-complete tree decomposition of $G[B_{T_v}\setminus\{y\}]$, and color $a$ does not appear on $B_{T_v}\setminus\{y\}$. We remove color $a$ from the set of $k$ available colors. Since $\alpha$ is $(V \setminus B_u)$-coherent, it is in particular $(B_{T_v}\setminus B_v)$-coherent. Take a color $b$ that does not appear on $B_v$. We apply the induction hypothesis on $T_v[B_{T_v}\setminus\{y\}]$ with $\ell-1$, $k-1$ and color $b$. Every vertex of $T_{B_v}\setminus B_v$ has been recolored at most once, and by assumption, neither color $a$ nor $b$ appears on this set. We recolor $y$ in $b$. Since $y \in B_v$, every vertex of $T_{B_v}$ has been recolored at most once, and color $a$ does not appear in $T_v$.
\end{itemize}

By Remark~\ref{rem:merged}, the resulting coloring is \textit{de facto} $(V\setminus B_u)$-coherent. Every vertex has been recolored at most once, no vertex of $B_u$ has been recolored, and no vertex is colored in $a$.
\end{proof}

\subsection{Obtaining a $V$-coherent coloring}\label{sec:obtainvcoherent}
In order to prove Theorem~\ref{thm:tw}, Lemma~\ref{compatibleson} ensures that we just have to transform any coloring into a $V$-coherent coloring in at most $n^2$ steps. We will consider the vertices one after the other and we will try to obtain a coloring that is coherent when restricted to the vertices already treated. The recoloring algorithm is detailed in the following lemma:

\begin{lemma}\label{lem:mainlemma}
Let $T$ be a $tw(G)$-complete tree decomposition. For every $\ell$-coloring $\alpha$ of $G$, there is a $V$-coherent coloring $\gamma_\alpha$ such that $d(\alpha,\gamma_{\alpha}) \leq n^2$.
\end{lemma}
\begin{proof}
The proof consists in a recoloring algorithm. We treat vertices one after the other, considering vertices that have at most one parent not yet treated. In other words, we treat babies of the remaining tree-decomposition. Our invariant will ensure that, when $X$ is treated, the current coloring is $X$-coherent. When a new vertex $x$ is treated, we just have to transform the current coloring in order to obtain a $(X \cup \{ x \})$-coherent coloring. At the end of the procedure, the whole vertex set is treated, and then the current coloring is $V$-coherent.

Let us now describe more formally the invariants. The set $F_i$ represents treated vertices at step $i$. Initially, no vertex is treated, so $F_0=\emptyset$. The coloring $c_i$ is the current $\ell$-coloring at the end of step $i$. Initially the coloring is $\alpha$, so $c_0=\alpha$. The invariants at the end of step $i$ are:
\begin{enumerate}[(i)]
\item\label{p1} $F_{i-1} \subset F_i \subseteq V$, and $|F_i|=i$.
\item\label{p2} $T[V \setminus {F_i}]$ is a $\min(tw(G),|V\setminus F_i|)$-complete tree decomposition of $G\setminus F_i$.
\item\label{p3} $c_i$ is an $\ell$-coloring of $G$ obtained from $c_{i-1}$ by recoloring vertices of $F_i$ at most twice.
\item\label{p4} $c_i$ is $F_i$-coherent.
\end{enumerate}

We proceed iteratively on $i$ from $1$ to $n$. Let $u$ be a leaf of $T[V \setminus F_i]$ and $x$ be a baby contained in $B_u$. In Fig.~\ref{fig:completedecomp}, the node $u$ is a leaf of $T[V \setminus \{ x_4,x_5,x_6,x_7,x_9 \}]$ and $x_3$ is the corresponding baby. We want to add $x$ in $F_i$. Denote by $F_{i+1}$ the set $F_i \cup x$. By Remark~\ref{rem:lcomplete} and since $x$ is a baby, $T[V \setminus F_{i+1}]$ is a complete tree decomposition. Thus (\ref{p1}) and (\ref{p2}) are immediately verified. The following consists in proving (\ref{p3}) and (\ref{p4}).


A \emph{residual component} of $T[V\setminus F_i]$ is a connected component of $T \setminus T[V \setminus F_i]$. Informally, a residual component is a subtree of the tree decomposition containing at least one treated vertex. A \emph{residual component of $u$} is a residual component containing a node adjacent to $u$. Note that vertices which appear in a bag of such a residual component are included in $F_i \cup B_u$.
In Fig.~\ref{fig:completedecomp}, subtrees $T_1^*$ and $T_2^*$ are the residual components of $u$ in $T[V \setminus \{ x_4,x_5,x_6,x_7,x_9 \}]$.

Let $F$ be the union of the residual components on $u$. And let $T^*$ be the subtree induced by $F \cup u$ rooted on $u$. Let us consider the graph $G'$ restricted to the vertices of $T^*$. Let $a$ be a color which does not appear in $B_u$. Note that the coloring $c_i$ restricted to $G'$ is $(V(G') \setminus B_u)$-coherent. Indeed, the vertices of $V(G') \setminus B_u$ are in $F_i$, and the coloring $\alpha$ is $F_i$-coherent. 

By applying Lemma~\ref{pfclaim1}, $c_i$ can be transformed into a $(B_{T^*}\setminus B_u)$-coherent coloring of $G'$ where no vertex is colored with $a$. Every vertex of $F_i$ is recolored at most once. Note that since vertices of $B_u$ are not recolored, the obtained coloring is proper on the whole graph.
Since $B_u$ is still in the clique tree, it means that no vertex of $B_u$ is in $F_i$. So if two vertices of $F_i$ are parents, either they are both in $B_{T^*}$ or none is in $B_{T^*}$. So the resulting coloring is $F_i$-coherent.

At this point, the color $a$ has disappeared from $T_u$ and the coloring is $V \setminus F_i$-coherent. So all the members of the family of $x$ that are in $F_i$ can be recolored with $a$, as the vertex $x$ itself. Every vertex is recolored at most once. 
Finally every vertex is recolored at most twice. So the resulting coloring $c_{i+1}$ satisfies conditions (\ref{p3}) and (\ref{p4}).

This operation is repeated until $F_i=V$, that is, $i=|V|$. When the last vertex is treated the coloring is $V$-coherent by (\ref{p4}). It follows from (\ref{p3}) that to recolor $G$ from $\alpha$ to $\gamma_\alpha=c_{n}$, it suffices to recolor each vertex $x$ at most $2\cdot (n-i+1)$, where $i$ is the smallest such that $x \in F_i$. Thus, on the whole, it suffices to make $2\cdot \frac{n(n+1)}{2}=n^2+2n$ recolorings.
Actually, the analysis can be slightly improved. Indeed, the vertex $x_i$ treated at step $i$ is recolored at most once (since vertices of $B_u$ are not recolored in Lemma~\ref{pfclaim1}). Therefore, every vertex is recolored at most $1+2 \cdot(n-i)$ times, which finally ensures that $d(\alpha,\gamma_\alpha) \leq n^2$.
\end{proof}

Note that the proof is totally algorithmic and runs in polynomial time (for a fixed number of colors).


\section{Cographs and distance-hereditary graphs}\label{sect:CDHproof}

In this Section, the notion of modules will be essential. Given a graph $G$, a subset $X$ of vertices is a \emph{module} if, for every vertex $y \notin X$ then $y$ is adjacent to either every vertex of $X$ or none of them. A subset $X$ of vertices is a \emph{strong module} if $|X|\geq 2$ and for every module $M$ of $G$, $M$ and $X$ are either disjoint or contained one in the other. In this section, the total order on the colors is the standard order on the integers.

\begin{remark}\label{rem:threemodules}
For every graph, any three strong modules $M_1, M_2$ and $M_3$ such that $M_1 \subsetneq M_2 \subsetneq M_3$ satisfy $\chi(M_1)>\chi(M_3)$.
\end{remark}

\begin{proof}
Assume by contradiction that three strong modules $M_1,M_2,M_3$ such that $M_1\subsetneq M_2 \subsetneq M_3$ satisfy $\chi(M_1)=\chi(M_2)=\chi(M_3)$. There is no edge between any vertex $x$ of $M_3 \setminus M_1$ and $y \in M_1$. Indeed otherwise $x$ must be connected to every vertex of $M_1$ since $M_1$ is a module.
And thus $\chi(M_3) \geq \chi (G[M_1 \cup x])= \chi(M_1)+1$, a contradiction. 

Therefore $M_3 \setminus M_1$ is a module of $G$. Indeed, they have the same neighborhood in $V \setminus M_3$ since $M_3$ is a module and the same neighborhood in $M_1$ since there is no edge beween $M_3 \setminus M_1$ and $M_1$. Though, $M_3 \setminus M_1$ strictly intersects with $M_2$, a contradiction since $M_2$ is a strong module.
\end{proof}

\begin{remark}\label{rem:moduledepth}
For every graph $G$, every vertex $x$ belongs to at most $2 \cdot \chi(G)$ distinct strong modules.
\end{remark}

\begin{proof}
Since every strong module containing $x$ intersects any other on $x$, all of them are included the ones into the others by definition of strong modules. Since the chromatic number of a module is at most $\chi(G)$, Remark~\ref{rem:threemodules} ensures that $x$ is contained in at most $2 \cdot \chi(G)$ strong modules.
\end{proof}

\subsection{Cographs and quasi-cographs}

A graph $G=(V,E)$ is a \emph{cograph}~\cite{lerchs71} if it does not contain induced paths of length $4$. Equivalently, a graph is a cograph if:
\begin{itemize}
 \item $G$ is a single vertex.
 \item Or $V$ can be partitionned into $V_1,V_2$ such that $G[V_1]$ and $G[V_2]$ are cographs and there is no edge between any vertex of $V_1$ and any vertex of $V_2$.
 \item Or $V$ can be partitionned into $V_1,V_2$ such that $G[V_1]$ and $G[V_2]$ are cographs and every vertex of $V_1$ is adjacent to every vertex of $V_2$ (such an operation is called a \emph{join}).
\end{itemize}
Note in particular that in the two last constructions, both subsets $V_1$ and $V_2$ are modules. In addition, in the second (resp. third) construction, $V_1$ is a strong module of $G$ iff it is not the disjoint union (resp. join) of two cographs. It follows that in a cograph, all the strong modules are a single vertex or the disjoint union of strong modules, or the join of strong modules. The cographs have been introduced in 1971 and have been extensively studied (see~\cite{brandstadt99}).

A coloring $c$ of a cograph $G$ is said to be \emph{modular} if for every strong module $M$ of $G$ that is the disjoint union of $p$ strong modules $M_1,\ldots,M_p$, the coloring $c$ uses a set $S$ of exactly $\chi(G[M])$ colors on the set $M$, and every $M_i$ is colored in $c$ with the first $\chi(G[M_i])$ colors in $S$. \emph{Partially joining} a clique $D$ to a graph means taking a non-empty subset of $D$ and joining it to the graph.

A \emph{quasi-cograph rooted in $H$}, where $H$ is a cograph, is the graph obtained by partially joining $p$ cliques $D_1,\cdots,D_p$ to $p$ modules $C_1,\cdots,C_p$ of $H$ such that no two $C_i$'s strictly overlap. By partially joining, we mean that some vertices of the clique are totally connected to the module and others are not connected.
 Note that $p$ may be equal to $0$, so cographs are in particular quasi-cographs rooted in themselves. Note also that two $C_i$'s might be identical. In other words, a quasi-cograph is obtained from a cograph by adding pending cliques on non-strictly overlapping modules of the cograph $H$ (see Figure~\ref{fig:quasico}). Given a subset $X$ of vertices of $H$, the \emph{vertices rooted in $X$} are the vertices of the quasi cograph which are not in $H$ and connected to a vertex of $X$. By extension, a modular coloring of a quasi-cograph is a coloring of it whose restriction to the underlying cograph is modular.
 
 \begin{figure}[!h]
 \centering
\begin{tikzpicture}[scale=1.5]
\tikzstyle{blacknode}=[draw,circle,fill=black,minimum size=6pt,inner sep=0pt]

\draw (0,0) node[blacknode] (u) {}
-- ++ (-60:1cm) node[blacknode] (v1) {};

\draw (u)
-- ++ (180+60:1cm) node[blacknode] (v2) {};

\draw (-0.5,-2) node[blacknode] (w1) {};
\draw (-0.5,-3) node[blacknode] (w3) {};
\draw (0.5,-2) node[blacknode] (w2) {};
\draw (0.5,-3) node[blacknode] (w4) {};

\draw (v2)
-- ++ (180+45:1cm) node[blacknode] (x1) {};
\draw (x1)
-- ++ (-90:1cm) node[blacknode] (x2) {};
\draw (x1)
-- ++ (180:1cm) node[blacknode] (x3) {};

\draw (x2) edge node {} (x3);

\foreach \i in {1,2,3}
\foreach \j in {2,3,4}
\draw (w\i) edge node {} (w\j);

\foreach \i in {2,3}
\draw (x\i) edge node {} (v2);

\foreach \i in {1,2}
\foreach \j in {1,2}
\draw (w\i) edge node {} (v\j);

\draw[thick,loosely dotted] (0,-0.5) circle [x radius=1cm, y radius=0.75cm];

\draw[thick,densely dotted] (-0.5,-0.85) circle [x radius=0.25cm, y radius=0.25cm];
\draw[thick,densely dotted] (0,-0.85) circle [x radius=1cm, y radius=0.25cm];
\end{tikzpicture}
  \caption{An example of a quasi-cograph $G$. The three top vertices form a cograph $H$, and the graph $G$ can be seen as a quasi-cograph rooted in $H$, with $C_i$'s as circled in the figure. Note that there may be multiple choices for $H$.}
  \label{fig:quasico}
 \end{figure}
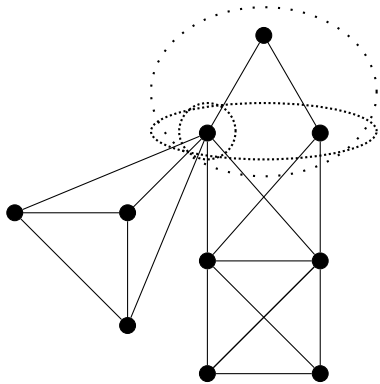

In order to prove the bounds on the mixing numbers of cographs and distance-hereditary graphs, we will first prove a few recoloring properties on quasi-cographs. Before stating the recoloring results, let us recall interesting properties of modules in cographs.

\begin{claim}\label{cl:trivialstrong}
A cograph $C$ that admits no non-trivial strong module is either a clique or a stable set.
\end{claim}
\begin{proof}
By induction on $|C|$.

\begin{itemize}
\item If $C$ is a single vertex, the result holds.
\item If $C$ is the disjoint union of two cographs $C_1$ and $C_2$. A non-trivial strong module of $C_1$ (resp. $C_2$) is a strong module of $C$, thus by the induction hypothesis, both $C_1$ and $C_2$ are either a clique or a stable set. Since $C$ is not a stable set, assume w.l.o.g. that $C_1$ is a clique (of size at least $2$). Then $C_1$ is a strong module of $C$ (indeed a module containing a strict subset of $C_1$ does not contain a vertex of $C_2$), a contradiction.
\item If $C$ is the join of two cographs $C_1$ and $C_2$. Similarly, we can assume that $C_1$ is a stable set, which implies that $C_1$ is a strong module of $C$, a contradiction. 
\end{itemize}
\end{proof}

We start with an easy lemma:
\begin{lemma}\label{lem:quasioptimal}
Let $G=(V,E)$ be a quasi-cograph rooted in $H$ and $k \geq \chi(G)+1$. Every $k$-coloring $\alpha$ of $G$ can be recolored into a modular coloring $\gamma_\alpha$ with $\gamma_\alpha(H) \subset \alpha(H)$ by recoloring every vertex at most $4 \cdot k \cdot \chi(G)$ times, while using no color on $H$ beside $\alpha(H)$.
\end{lemma}
\begin{proof}
Then $G$ is a graph obtained by partially joining $p$ cliques $D_1,\cdots,D_p$ to $p$ modules $C_1,\cdots,C_p$ of $H$ such that no two $C_i$'s strictly overlap.

Let $M$ be a smallest strong module of $H$ such that $H[M]$ is not modularly colored. 
Thus every maximal strong module of $H$ strictly included in $M$ is modularly colored: we consider these strong modules to behave as a clique (where each color class corresponds to a single vertex) which is possible by Remark~\ref{rem:merged} (all the vertices of the same color class are merged which ensure that any recoloring is still a modular coloring for these modules). If $M$ was the join of $p$ strong modules, it is now a clique and its coloring is necessarily modular. Then $M$ was the disjoint union of cliques $K_1, \ldots, K_p$. 
Since $M$ is a strong module, every $C_j$ that intersects $M$ but does not contain it is strictly included in $M$. Since $C_j$ is a module, if $C_j$ is strictly included in $M$ then $C_j$ is either strictly included in some $K_i$, or is a union of $K_i$'s.

Let $C$ be a minimal $C_i$ included in $M$ that is not minimally colored, or $M$ itself if there is no such $C_i$. The $C_i$'s strictly included in $C$ are considered to behave as cliques (again, where each color class is merged into a single vertex). We denote $\{K'_i\}_i$ the resulting (at least two) disjoint cliques in $C$. Note that all the cliques partially joined to a strict subset of $C$ are partially joined to some $C_i$ strictly included in $C$ and thus to a subset of some $K'_j$. Let $\alpha_1,\alpha_2,\ldots,\alpha_{q}$ be the colors used on $C$, \emph{in that order} (i.e. with $\alpha_1 < \alpha_2 < \cdots < \alpha_q$). For every $j$ from $1$ to $\chi(H[M])$, for each $K'_i$ with $|K'_i|\geq j$, we pick a vertex $u_{i,j}$ that is colored with a color larger than $\alpha_{|K'_i|}$, if any.

We then recolor if needed the pending cliques on some subset of $K'_i$ in order to remove the color $\alpha_j$ from the neighborhood of $u_{i,j}$, as follows. For a given clique $K''$ partially joined to a subset $A$ of $K'_i$, since $K''$ is partially joined to $A$, let $B$ be the subset of $K''$ that is joined to $A$. If $u_{i,j}\not\in A$ or the color $\alpha_j$ does not appear on $B$, we do not recolor anything. If $u_{i,j} \in A$ and there is a vertex $v \in B$ whose color is $\alpha_j$. Then, since we have $k \geq \chi(G)+1$ and $A \cup B$ is a clique, there is a color $\gamma$ that does not appear on $A \cup B$. If there is a vertex $w$ of $K'' \setminus B$ that is colored in $\gamma$, we use the same argument to say that there is a color $\delta$ that does not appear on $K''$. Now we recolor the vertices as follows: $w$ in $\delta$, $v$ in $\gamma$, $u_{i,j}$ and $w$ in $\alpha_j$, and finally $v$ in the initial color of $u_{i,j}$ (we synchronize with other cliques partially attached to a subset of $K'_i$ in order to match the recoloring of $u_{i,j}$). Note that neither $v$ nor $w$ will be recolored at a future step: the color of $v$ is now greater than $\alpha_{|K'_i|}$ and that of $w$ now appears on $A$ (note that $u_{i,j}$ will not be recolored at a future step).

Eventually, when step $\chi(H[M])$ is done, the only vertices that were recolored are in $C$ or in a pending clique joined to a subset of $C$, and each was recolored at most twice.

We repeat this operation until $M$ is modularly colored. Since each vertex in $M$ that is recolored is recolored into a smaller color, each vertex is recolored at most $k$ times. Then the vertices in the pending cliques are recolored at most $2k$ times. 

Now we repeat these operations until $G$ is modularly colored. By Remark~\ref{rem:moduledepth}, every vertex will be recolored at most $4\cdot k \cdot \chi(G)$ times.
\end{proof}

It now remains to prove that a quasi-cograph can be recolored from any modular coloring to any other. We say that two colorings $c_1$ and $c_2$ of a graph \emph{agree up to translation} on a set $X$ if, for $\alpha_1,\cdots,\alpha_q$ the colors used by $c_1$ on $X$, and $\beta_1,\cdots,\beta_r$ the colors used by $c_2$, both in that order, we have $q=r$ and for any vertex $u \in X$, it holds that $c_1(u)=\alpha_i$ iff $c_2(u)=\beta_i$.

\begin{lemma}\label{lem:quasicographmodular}
Let $G=(V,E)$ be a quasi-cograph rooted in $H$ and $k\geq \chi(G)+1$. For any two modular colorings $\gamma_1$ and $\gamma_2$ of $G$, and any color $a \not\in \gamma_1(H)$, there exists a modular coloring $c$ of $G$ such that $c(H)=\gamma_1(H)$, $\gamma_2$ and $c$ agree up to translation on $V$, and $G$ can be recolored from $\gamma_1$ to $c$ by recoloring each vertex at most $6 \cdot \chi(G)$ times, while using no color on $H$ beside $\{a\}\cup \gamma_1(H)$.
\end{lemma}
\begin{proof}
Then $G$ is a graph obtained by partially joining $p$ cliques $D_1,\cdots,D_p$ to $p$ modules $C_1,\cdots,C_p$ of $H$ such that no two $C_i$'s strictly overlap.
We build step-by-step a modular coloring $c$ of $G$ that agrees up to translation with $\gamma_2$ on $V$ and does not use $a$ on the $D_i$'s. We consider $G$ to be initially colored in $\gamma_1$. We recolor if needed the $D_i$'s so as to remove the color $a$ from them: this can be done by recoloring each vertex at most once. We set $c_0$ to be the resulting coloring.

At step $j$, let $M$ be a smallest strong module on which $c_j$ and $\gamma_2$ do not agree up to translation. Remember that $c_j$ and $\gamma_2$ are modular colorings, thus if $M$ is the disjoint union of strong modules $N_1, \ldots, N_p$, then the colors used by $c_j$ on each $N_i$ are exactly the $\chi(H[N_i])$ smallest colors used by $c_j$ on $M$, and the same stands for $\gamma_2$. Combined with the fact that $c_j$ and $\gamma_2$ agree up to translation on each $N_i$, it yields that $c_j$ and $\gamma_2$ necessarily agree up to translation on $M$. Then $M$ cannot be the disjoint union of strong modules, nor can it be a single vertex, so $M$ must be the join of $p$ strong modules $N_1, \ldots, N_p$. By assumption, $c_j$ and $\gamma_2$ agree up to translation on each $N_i$. Then we consider each $N_i$ to behave just as a clique (each color class corresponds to a single vertex). Now $M$ is a clique. Let $d$ be the coloring of $M$ that is the translation of $\gamma_2$ to the color set used by $c_j$ on $M$, while maintaining the order on the colors (i.e. the vertex of $M$ colored with the smallest color in $\gamma_2$ is colored in the smallest color in $d$). By Lemma~\ref{lemma:clique}, we can recolor $M$ from $c_j$ to $d$ by recoloring each vertex of $M$ at most twice, while using no new color beside $a$. We recolor as necessary the $D_i$'s joined to some $C_i \subseteq M$ in the meanwhile, and then remove if necessary the color $a$ from them: each vertex need only be recolored at most three times. Let $c_{j+1}$ be the resulting coloring of $G$. 

By Remark~\ref{rem:moduledepth}, each vertex belongs to at most $2\chi(G)$ strong modules. Then we are sure to obtain a modular coloring that agrees up to translation with $\gamma_2$ on $V(H)$ at step $2\chi(G)$ at most. Consequently, recoloring each vertex $6 \chi(G)$ times is enough. Then we recolor as necessary the cliques partially joined to $H$. This can be done by recoloring each vertex at most twice.
\end{proof}

Let us now use it to prove Theorem~\ref{thm:distancehereditary}. 

\subsection{Distance-hereditary graphs}\label{sec:disthered}

Two vertices $x,y$ are \emph{false twins} if $N(x) = N(y)$. In particular $xy$ is not an edge. By symmetry, two vertices $x,y$ are \emph{true twins} if $N(x) \cup x = N(y) \cup y$.
Let us first define the distance-hereditary graphs. A graph is \emph{distance-hereditary} if for every pair of vertices $x,y$ and for every pair of paths $P_1$ and $P_2$ from $x$ to $y$, the length of $P_1$ equals the length of $P_2$. 
\begin{theorem}\cite{HammerM90}\label{theorem:HammerM}
A distance-hereditary graph can be built from a unique vertex graph by operations of false twins, true twins and adding a pendant vertex.
\end{theorem}
An order for which, at every step a vertex can be added by one of these three operations is called a \emph{construction order}.
Remark that a cograph is a graph in which every connected component is a distance-hereditary graph with diameter at most 2~\cite{brandstadt99}.

Given a rooted tree $T$ (i.e. a tree whose edges are oriented and such that there exists a node $r$ such that any node can be reached by a path starting on $r$), for any two nodes $u,v$ of $T$, we say $v$ is the \emph{son} of $u$ if there is an arc from $u$ to $v$ in $T$. The node $u$ is a \emph{leaf} if $u$ has no son. 

\begin{lemma}\label{lemme:treeparthereditary}
Every distance-hereditary graph $G=(V,E)$ admits an oriented tree $T$ such that:
\begin{itemize}
 \item A bag $B_u$ is associated to every node $u$ of $T$.
 \item Every vertex of $G$ appears in exactly one bag.
 \item For every node $u$, $G[B_u]$ is a cograph.
 \item For every arc $uv$ of $T$, there exists a module $M_{u,v}$ of $G[B_u]$ such that $B_v$ is completely (and only) connected to $M_{u,v}$ and such that for any two sons $v,w$ of $u$, $M_{u,v}$ and $M_{u,w}$ are either disjoint or included one in the other.
\end{itemize}
\end{lemma}
\begin{proof}
This tree can be inductively built out of the construction order as follows. When the graph contains only one vertex, the tree contains only one node, which corresponds to the only vertex. At each step, if the new vertex is the twin of some vertex $u$, then it is assigned to the node where $u$ belongs, and if it is a pendant vertex of $u$, a new node is created that contains only the new vertex, and is incorporated in the tree as the son of the node where $u$ belongs. Since a node is build via operations of true and false twins, each node is a cograph.
\end{proof}

\begin{remark}\label{rem:tree2dh}
Reciprocally, every graph for which such a tree exists is a distance-hereditary graph. 
\end{remark}
\begin{proof}[Sketch of the proof]
Let us exhibit the complement of a construction order. If a leaf contains at least two vertices, then delete one of them using false or true twin operations. Consider a leaf $x$ of the internal tree and let $M$ be a minimum module (i.e. no other son of $x$ has its module strictly included in $M$). If $M$ contains at least two vertices, one can delete one of them using twins operations. If $M$ is a single vertex, then the vertex of the leaf is a pending vertex.
\end{proof}

For any distance-hereditary graph $G$, a \emph{modular coloring} is a coloring of $G$ such that for every node $u$ of $T(G)$, the cograph $G[B_u]$ is modularly colored.

\begin{proof}[Proof of Theorem~\ref{thm:distancehereditary}.]

Let $G$ be a connected distance-hereditary graph, and let $k \geq \chi(G)+1$.
We prove Theorem~\ref{thm:distancehereditary} in two steps, as in the case of Theorem~\ref{thm:tw}. We first introduce a few definitions.

Given a node $u$ of $T(G)$, we denote by $H_u(G)$ the graph induced in $G$ by the union of $B_v$'s such that $v$ is a node of the subtree of $T(G)$ rooted in $u$.
Given a node $u$ of $T(G)$, we say that $u$ is a \emph{quasi-leaf} if $H_u(G)$ is a quasi-cograph. Note that a node $u$ of $T(G)$ is a quasi-leaf if $u$ is a leaf, or if $u$ is a node whose every son is either a leaf whose bag induces a clique or a node $v$ with a single son $w$ such that both $B_v$ and $B_w$ induce a clique (in this case the clique $B_v \cup B_w$ is partially attached to $H_u(G)$).
A \emph{maximal} quasi-leaf in $T(G)$ is a quasi-leaf $u$ that is the son of no quasi-leaf, and whose bag does not induce a clique or that is neither a leaf nor a node with a single son whose bag induces a clique.

Let $\ell(G)$ be the number of nodes in $T(G)$ that are neither leaves nor non-maximal quasi-leaves.

\begin{lemma}\label{lem:2modular}
For any $k$-coloring $\alpha$ of $G$, there exists a modular coloring $\gamma_\alpha$ of $G$ such that $d(\alpha,\gamma_\alpha) \leq (4 \cdot k \cdot \chi(G)+2) \cdot n^2$. In particular, $\gamma_\alpha$ uses exactly $\chi(G)$ colors on $G$. 
\end{lemma}
\begin{proof}
Let $G_0$ be $G$. If $\ell(G)=0$ then $G$ is a quasi-cograph and we apply Lemma~\ref{lem:quasioptimal} then recolor if needed the cliques that are partially joined to the cograph in order to use exactly $\chi(G)$ colors on $G$. 

At each step $i$ and until $\ell(G_i)=0$, we merge some vertices in $G_i$ to obtain a distance-hereditary graph $G_{i+1}$ such that $\ell(G_{i+1})<\ell(G_i)$ and a modular coloring of $G_{i+1}$ yields a modular coloring of $G_i$. We simultaneously build a coloring $c_i$ of $G_i$ such that $c_{i+1}$ is obtained from $c_i$ by recoloring each vertex at most $(4 k \cdot \chi(G)+2)$ times. We proceed as follows.

Since $\ell(G_i)>0$, the oriented tree $T(G_i)$ contains a node $u$ that is a maximal quasi-leaf of $T(G_i)$. The graph $H_u(G_i)$ is a quasi-cograph. We apply Lemma~\ref{lem:quasioptimal} to $H_u(G_i)$ in order to obtain a modular coloring of $G_i[B_u]$ without recoloring any vertex outside $H_u(G_i)$ and by recoloring each vertex at most $4 k \chi(G)$ times (and by never recoloring a vertex of the bag of $u$ with a color which is not initially in $c_i(B_u)$). We merge the vertices in $B_u$ according to their color classes. Note that at this point, the set $B_u$ and then in particular every module of $H_u$ is a clique, so any recoloring of $H_u$ is still a modular coloring.

\begin{figure}
 \centering
 \includegraphics[scale=0.8]{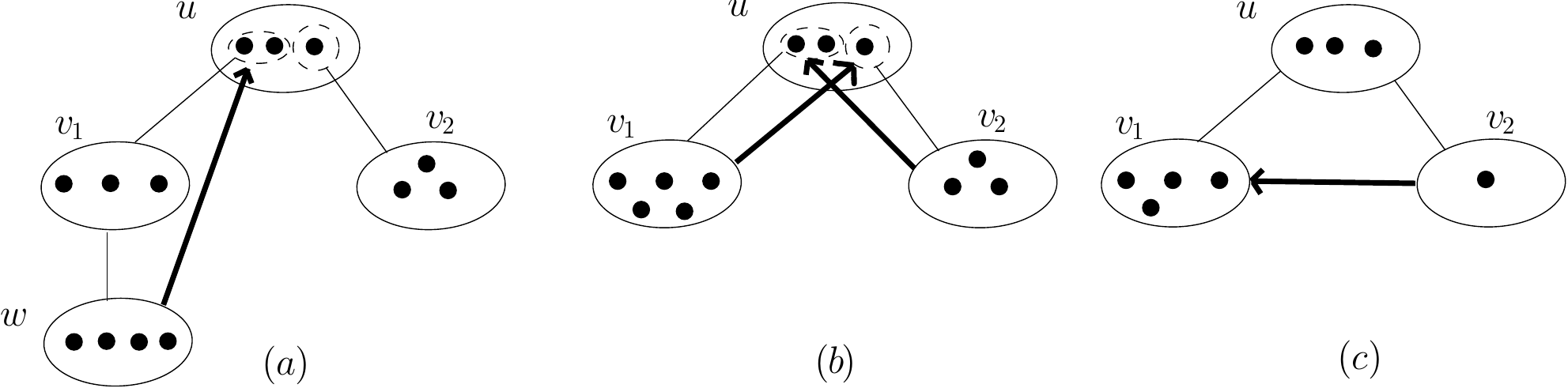}
 \caption{Three steps of the recoloring of $H_u(G_i)$. The dotted parts of $u$ correspond to modules on which $v_1$ and $v_2$ are attached. Arrows illustrate the recoloring operations done at each step.}
 \label{fig:dh}
\end{figure}

This proof is illustrated in Figure~\ref{fig:dh}. Let us show that, by recoloring each vertex at most twice, we can assume that the node $u$ has at most one son, and that this son is a clique which is joined to $B_u$. Let $v$ be a son of $u$. By assumption, $v$ has at most one son $w$. First recolor, as long as possible, some vertices of $B_w$ (if they exist) with one color of the clique $K$ (recall that $B_w$ is not adjacent to $B_u$), see Figure~\ref{fig:dh}(a). Once the set of colors used on $B_W$ and the set of colors used on $K$ are included one in the other, we can merge the vertices of $B_w$ with the corresponding vertex in $K$. Note that the graph is still a quasi-cograph where the remaining vertices of $B_w$, if any, can be put in $B_v$ (due to the merging, the vertices of $B_w$ are now adjacent to all the vertices of $K$). Thus we can consider that $B_v$ has no son, see Figure~\ref{fig:dh}(b). Next, we recolor the vertices of $B_v$ with colors of $B_u$ as long as it is possible. Once again we merge the vertices of $B_v$ and $B_u$ that have the same color, and we obtain a clique on which other cliques have been attached, see Figure~\ref{fig:dh}(c). We repeat this operation for every son $v_i$ of $u$. Until now, every vertex is recolored at most once. By recoloring every vertex of the $B_{v_i}$'s at most once, we can assume that the set of colors used on each clique attached on $B_u$ is a subset of the largest one. Then we can again merge the vertices of the same color in order to be sure that only one clique is attached on $B_u$. Throughout this process, every vertex has been recolored at most twice.

Let $G_{i+1}$ be the resulting distance-hereditary graph, and $c_{i+1}$ be the resulting coloring of $G_{i+1}$. Thus, after this step, $B_u$ is a clique and there is at most one leaf (which induces a clique) attached to $u$ and such a leaf, if it exists is a clique. So $u$ is a non-maximal quasi-leaf of $T(G_{i+1})$, so $\ell(G_{i+1})<\ell(G_i)$. Besides, $c_{i+1}$ is obtained from $c_i$ by recoloring each vertex at most $4 k \cdot \chi(G)+2$ times.

Since the number of vertices decreases at each step, we reach a step $j$ such that $\ell(G_j)=0$ in at most $n$ steps. The graph $G_j$ is a clique, whose coloring cannot be anything but modular, hence the result.
\end{proof}

\begin{lemma}\label{lem:modular2}
For any two modular colorings $\gamma_1$ and $\gamma_2$ of $G$ with $|\gamma_2(G)|=\chi(G)$, it holds that $d(\gamma_1,\gamma_2) \leq (6 \cdot \chi(G)+4)\cdot n^2$. 
\end{lemma}
\begin{proof}
The proof of this lemma is rather similar to that of Lemma~\ref{lem:2modular}.

Let $G_0$ be $G$. If $\ell(G)=0$ then Lemma~\ref{lem:modular2} holds by Lemma~\ref{lem:quasicographmodular}. At each step $i$ and until $\ell(G_i)=0$, we merge some vertices in $G_i$ to obtain a distance-hereditary graph $G_{i+1}$ such that $\ell(G_{i+1})<\ell(G_i)$ and a modular coloring of $G_{i+1}$ that agrees up to translation with $\gamma_2$ yields a modular coloring of $G_i$ that agrees up to translation with $\gamma_2$. We simultaneously build a modular coloring $c_i$ of $G_i$ such that $c_{i+1}$ is obtained from $c_i$ by recoloring each vertex at most $(4 k \cdot \chi(G)+2)$ times. We proceed as follows.

Since $\ell(G_i)>0$, the oriented tree $T(G_i)$ contains a node $u$ that is a maximal quasi-leaf of $T(G_i)$. The graph $H_u(G_i)$ is a quasi-cograph. Since $c_i$ is a modular coloring, we know that for $I$ the neighborhood of $H_u(G_i)$ in $G\setminus (H_u(G_i))$, the number of colors used by $c_i$ on $I \cup B_u$ is at most $\chi(G)$, thus there is a color $a$ that is not used on $I \cup B_u$. We apply Lemma~\ref{lem:quasicographmodular} to $H_u(G_i)$ with color $a$ in order to obtain a modular coloring of $G_i[B_u]$ without recoloring any vertex outside $H_u(G_i)$ and by recoloring each vertex at most $4 k \chi(G)$ times (and by never recoloring a vertex of the bag of $u$ with a color which is not in $c_i(B_u) \cup \{a\}$). We merge the vertices in $H_u(G_i)$ according to their color classes (note that the modular coloring obtained by Lemma~\ref{lem:quasicographmodular} agrees up to translation with $\gamma_2$ on $H_u(G_i)$ thus uses at most $\chi(G)$ colors). Thus we obtain a coloring $c_{i+1}$ that allows us to merge $G_i$ into a distance-hereditary graph $G_{i+1}$ where $u$ has at most one son $v$ in $T(G_{i+1})$, and the bag of $v$ is a clique.

Let $G_{i+1}$ be the resulting distance-hereditary graph, and $c_{i+1}$ be the resulting coloring of $G_{i+1}$. Note that, after this step, $u$ is a leaf of $T(G_{i+1})$, so $\ell(G_{i+1})<\ell(G_i)$. Besides, $c_{i+1}$ is obtained from $c_i$ by recoloring each vertex at most $4 k \cdot \chi(G)+2$ times.

Since the number of vertices decreases at each step, we reach a step $j$ such that $\ell(G_j)=0$ in at most $n$ steps. The graph $G_j$ is a clique, whose coloring cannot be anything but modular, hence the result.
\end{proof}
\end{proof}

\section{Further work}\label{sec:ccl}

Graphs of treewidth at most $k$ are $k$-degenerate graphs. The $(k+2)$-recoloring diameter of $k$-degenerate graphs at most $2^n$~\cite{Cereceda}. Note that the bound on the number of colors is optimal since $K_n$ is $(n-1)$-degenerate. Does the class of $k$-degenerate graphs have a polynomial $(k+2)$-recoloring diameter? Or, a weaker question, can we obtain a polynomial recoloring diameter when the number of color increases?

This question seems very challenging. The class of $k$-degenerate graphs also contains some sub-classes that are themselves interesting. One of the most famous is the class of planar graphs (which are $5$-degenerate).
\begin{conjecture}
For any planar graph $G$ and any integer $k$, if $k \geq 7$ then $R_k(G)$ has a polynomial diameter.
\end{conjecture}
This bound of $7$ would be optimal since there are planar graphs that are not $5$-mixing (see Figure~\ref{fig:mix5}) or not $6$-mixing (see Figure~\ref{fig:mix6}) (note that such examples also appear in Figure 2.2 of~\cite{Cereceda}).

\begin{figure}[!h]
\centering
\parbox[b]{3in}{
\centering
 \includegraphics[scale=0.7,bb=0 0 100 100]{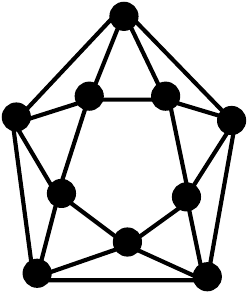}
 \caption{A planar graph that is not 5-mixing.}
 \label{fig:mix5}
}
\qquad
\parbox[b]{3in}{
\centering
 \includegraphics[scale=0.7,bb=0 0 100 100]{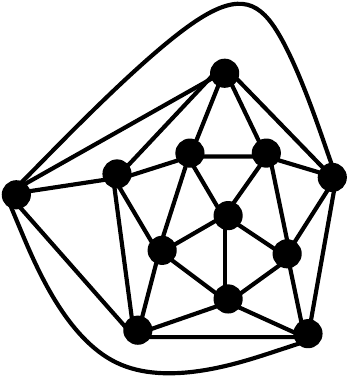}
 \caption{A planar graph that is not 6-mixing.}
 \label{fig:mix6}
 }
\end{figure}
Note that outerplanars graphs have a quadratic recoloring diameter since they have treewidth at most $2$. The quadratic lower bound is optimal~\cite{BonamyJ12} (see Figure~\ref{fig:outerplanar}).

\begin{figure}[!h]
\centering
 \includegraphics[bb=0 0 100 100]{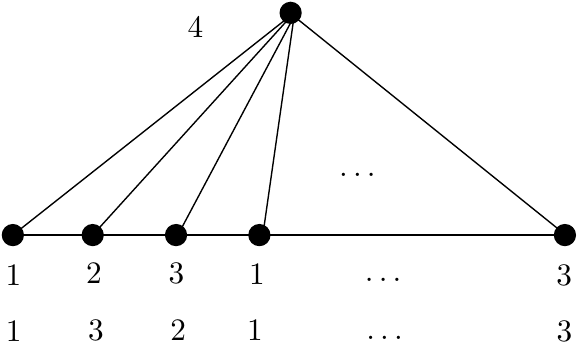}
 \caption{A $3$-colorable outerplanar graph which has a quadratic recoloring diameter.}
 \label{fig:outerplanar}
\end{figure}

Another interesting point is the existence of a hamiltonian cycle in the recoloring graph. In other words, is it possible to find a sequence of distinct recolorings which contains all the propers colorings and such that the consecutive colorings are adjacent? Consider for instance $2$-colorings of stable sets on $n$ vertices. The corresponding graph is the $n$-dimensional hypercube. Such graphs admit hamiltonian cycles, known as Gray codes. Gray codes, and their generalization, were extensively studied (see~\cite{Savage96} for a survey). Could there be some similar structure to be observed in other cases?

\bibliographystyle{plain}
\bibliography{biblitw}

\begin{thebibliography}{10}

\bibitem{brandstadt99}
Jeremy P.~Spinrad Andreas~Brandstadt, Van Bang~Le.
\newblock {\em Graph Classes}.
\newblock 1999.

\bibitem{ArnborgCP87}
S.~Arnborg, D.~Corneil, and A.~Proskurowski.
\newblock Complexity of finding embeddings in a k-tree.
\newblock {\em SIAM Journal on Algebraic Discrete Methods}, 8(2):277--284,
  1987.

\bibitem{Beyer82}
T.~Beyer, S.~M. Hedetniemi, and S.~T. Hedetniemi.
\newblock A linear algorithm for the grundy number of a tree.
\newblock In {\em Proceedings of the thirteenth southeastern conferenceon
  combinatorics, graph theory and computing}, 1982.

\bibitem{Bodlaender93}
Hans~L Bodlaender.
\newblock A linear time algorithm for finding tree-decompositions of small
  treewidth.
\newblock {\em Proceedings of the twentyfifth annual ACM symposium on Theory of
  computing}, 25(6):226--234, 1993.

\bibitem{BonamyB13}
M.~Bonamy and N.~Bousquet.
\newblock Recoloring bounded treewidth graphs.
\newblock In {\em to appear in LAGOS'13}, 2013.

\bibitem{BonamyJ12}
M.~Bonamy, M.~Johnson, I.~Lignos, V.~Patel, and D.~Paulusma.
\newblock Reconfiguration graphs for vertex colourings of chordal and chordal
  bipartite graphs.
\newblock {\em Journal of Combinatorial Optimization}, pages 1--12, 2012.

\bibitem{BonsmaC07}
P.~Bonsma and L.~Cereceda.
\newblock Finding paths between graph colourings: {PSPACE}-completeness and
  superpolynomial distances.
\newblock In {\em MFCS}, volume 4708 of {\em Lecture Notes in Computer
  Science}, pages 738--749, 2007.

\bibitem{Cereceda}
L.~Cereceda.
\newblock {\em Mixing Graph Colourings}.
\newblock PhD thesis, London School of Economics and Political Science, 2007.

\bibitem{Cereceda09}
L.~Cereceda, J.~van~den Heuvel, and M.~Johnson.
\newblock Mixing 3-colourings in bipartite graphs.
\newblock {\em Eur. J. Comb.}, 30(7):1593--1606, 2009.

\bibitem{CerecedaHJ11}
L.~Cereceda, J.~van~den Heuvel, and M.~Johnson.
\newblock Finding paths between 3-colorings.
\newblock {\em Journal of Graph Theory}, 67(1):69--82, 2011.

\bibitem{Christen79}
C.~Christen and S.~Selkow.
\newblock Some perfect coloring properties of graphs.
\newblock {\em Journal of Combinatorial Theory, Series B}, 27(1):49 -- 59,
  1979.

\bibitem{Diestel}
R.~Diestel.
\newblock {\em Graph Theory}, volume 173 of {\em Graduate Texts in
  Mathematics}.
\newblock Springer-Verlag, Heidelberg, third edition, 2005.

\bibitem{Gopalan09}
P.~Gopalan, P.~Kolaitis, E.~Maneva, and C.~Papadimitriou.
\newblock The connectivity of boolean satisfiability: Computational and
  structural dichotomies.
\newblock {\em SIAM J. Comput.}, pages 2330--2355, 2009.

\bibitem{Gyarfas88}
A.~Gyárfás and J.~Lehel.
\newblock Online and first-fit colorings of graphs.
\newblock {\em Journal of Graph Theory}, pages 217--227, 1988.

\bibitem{HammerM90}
P.~L. Hammer and F.~Maffray.
\newblock Completely separable graphs.
\newblock {\em Discrete Appl. Math.}, 27(1-2):85--99, 1990.

\bibitem{ItoD11}
T.~Ito, E.~Demaine, N.~Harvey, C.~Papadimitriou, M.~Sideri, R.~Uehara, and
  Y.~Uno.
\newblock On the complexity of reconfiguration problems.
\newblock {\em Theor. Comput. Sci.}, 412(12-14):1054--1065, 2011.

\bibitem{ItoD09}
T.~Ito, M.~Kamiński, and E.~Demaine.
\newblock Reconfiguration of list edge-colorings in a graph.
\newblock In {\em Alg. \& Data Struct.}, volume 5664 of {\em Lecture Notes in
  Computer Science}, pages 375--386. 2009.

\bibitem{Jerrum95}
M.~Jerrum.
\newblock A very simple algorithm for estimating the number of k-colorings of a
  low-degree graph.
\newblock {\em Random Structures \& Algorithms}, 7(2):157--165, 1995.

\bibitem{lerchs71}
H.~Lerchs.
\newblock On cliques and kernels.
\newblock Technical report, 1971.

\bibitem{Savage96}
Carla Savage.
\newblock A survey of combinatorial gray codes.
\newblock {\em SIAM Review}, 39:605--629, 1996.

\bibitem{Zaker06}
M.~Zaker.
\newblock Results on the grundy chromatic number of graphs.
\newblock {\em Discrete Mathematics}, 306(23):3166 -- 3173, 2006.

\end{thebibliography}

\end{document}